\documentclass[a4paper]{amsart}
\usepackage[margin=2.5cm]{geometry}
\usepackage{graphicx} % Required for inserting images
\usepackage{amsmath}
\usepackage{amssymb}
\usepackage{amsthm}
\usepackage[foot]{amsaddr}
\usepackage{color, soul, url}
\usepackage{xcolor}

\usepackage{tikz} % commutative diagrams
\usetikzlibrary{cd}

\newtheorem{theorem}{Theorem}

\begin{document}

\title{Self Organization in Computation \& Chemistry:\newline Return to AlChemy}

\author{Cole Mathis$^{1,2,*}$}
\author{Devansh Patel$^{1,3}$}
\author{Westley Weimer$^{4}$}
\author{Stephanie Forrest$^{1,2,3,5}$}

\address[1]{Biodesign Institute, Arizona State University, Tempe, AZ  85281}
\address[2]{School of Complex Adaptive Systems, Arizona State University, Tempe, AZ, 85281}
\address[3]{School of Computing and Augmented Intelligence, Arizona State University, Tempe, AZ 85281}
\address[4]{Electrical Engineering and Computer Science Department, University of Michigan, Ann Arbor, MI 48109}
\address[5]{Santa Fe Institute, Santa Fe, NM 87501}
\address[*]{cole.mathis@asu.edu}
\date{\today}

\maketitle

\begin{abstract}
How do complex adaptive systems, such as life, emerge from simple constituent parts? In the 1990s Walter Fontana and Leo Buss proposed a novel modeling approach to this question, based on a formal model of computation known as $\lambda$ calculus. The model demonstrated how simple rules, embedded in a combinatorially large space of possibilities, could yield complex, dynamically stable organizations, reminiscent of biochemical reaction networks. Here, we revisit this classic model, called AlChemy, which has been understudied over the past thirty years.  We reproduce the original results and study the robustness of those results using the greater computing resources available today.  Our analysis reveals several unanticipated features of the system, demonstrating a surprising mix of dynamical robustness and fragility. Specifically, we find that complex, stable organizations emerge more frequently than previously expected, that these organizations are robust against collapse into trivial fixed-points, but that these stable organizations cannot be easily combined into higher order entities. We also study the role played by the random generators used in the model, characterizing the initial distribution of objects produced by two random expression generators, and their consequences on the results. Finally, we provide a constructive proof that shows how an extension of the model, based on typed $\lambda$ calculus, \textcolor{black}{could simulate transitions between arbitrary states in any possible chemical reaction network, thus indicating a concrete connection between AlChemy and chemical reaction networks}. We conclude with a discussion of possible applications of AlChemy to self-organization in modern programming languages and quantitative approaches to the origin of life.
    
\end{abstract}

% \noindent \textbf{LEAD PARAGRAPH:} \textbf{How life emerged from simple molecular constituents is a longstanding question that has not been fully resolved. This paper resurrects a computational model called \textit{AlChemy} that was first proposed over thirty years ago to study how a set of simple computational rules could produce complex, dynamically stable organizations reminiscent of biochemical reaction networks.  Using the vastly increased computing resources available today, this paper reports on new computational experiments and analysis, which show that in AlChemy, stable and complex organizations emerge more frequently than previously expected and that these organizations are often robust against collapse. The paper also probes a key component of AlChemy, its method for generating random initial conditions, and finds that it is crucial to the model's success.  Finally, the paper gives a constructive proof showing that a slightly modified version of AlChemy can \textcolor{black}{simulate state transitions in any Chemical Reaction Network,} thereby establishing a closer connection between the model and the chemical reactions believed to have produced life.}

\section{Introduction}

The origin(s) of life remain an unresolved mystery in science, that is, how unconstrained reactive compounds were selected to form the organized biochemical networks that form the basis of  Darwinian lineages. In the early 1990s Walter Fontana and Leo Buss proposed a novel approach to this problem based on a formal model of computation known as the $\lambda$ calculus~\cite{fontana1993beyond, fontana1994arrival, fontana1994would}.  In a departure from the prevailing dynamical systems perspective, they modeled how novel objects might emerge through unconstrained interactions. This approach emphasized the primacy of chemical {\it reactivity} over predetermined {\it reactions} in their model. Their work was one of the first artificial chemistry models based on constructive processes, it focused on the production of new entities (construction) over time rather than motion through predetermined state spaces (dynamics)~\cite{dittrich2001artificial, kauffman1993origins}. The model, dubbed {\it AlChemy} (for Algorithmic Chemistry), was based on the construction and composition of $\lambda$ expressions---objects whose internal structure determines their relation and interaction with each other~\cite{fontana1996barrier}. 

The original analyses of AlChemy became foundational studies in understanding the origins of self-organized complexity~\cite{kauffman1993origins}. However, despite its large impact on several fields, from chemistry~\cite{quinkert1996aspects}, to economics~\cite{cederman2002endogenizing} and biology~\cite{wagner1996perspective}, the model itself has remained under-investigated. Two key reasons for this are: (i) skepticism about its applicability to chemical systems, or primitive living systems~\cite{szathmary1995classification}, and (ii) the technical challenge of systematically analyzing a system capable of such incredibly diverse behavior~\cite{krakauer2011challenges, fontana1996barrier, fontana1993beyond}.  Here we revive the AlChemy project using the original code base run on modern machines, considering their results in the context of a more systematic investigation using modern tools and computational resources. With nearly 30 years of progress in complex systems science and chemistry, we are also positioned to ask slightly different questions. In particular, we analyze the robustness of the original results by characterizing the stability of the organizations that form, how they respond to perturbation, and by performing a statistical survey of whether the identified organizations can be combined with each other.

The original code base for AlChemy has been faithfully hosted on the Santa Fe Institute's website, and is still accessible \cite{alchemy_site}, a remarkable achievement in open science and the only reason we were able to perform this reanalysis. Using the original code base, we analyze the statistical properties of the organizations produced by AlChemy rather than focusing on a few examples, as the original work did. We found results that are consistent, but not identical, to the original work, and there were a number of surprises. These differences suggested questions about the sensitivity of AlChemy to different methods of generating the initial soup of random $\lambda$ expressions.  \textcolor{black}{Finally, we return to the question of whether or not AlChemy is an appropriate model of chemistry, and show by construction that there exists a correspondence between the state transitions of any Chemical Reaction Network and the reaction sequences of a model like AlChemy. We use the term \textit{simulate} to refer to the correspondence between the states and state transitions of two different systems, without considering intermediate configurations or rates of reaction. The paper concludes by discussing how, after thirty years, AlChemy offers additional insight across several domains, including origin(s) of life, astrobiology, biology, and computer science.} 

\section{AlChemy}

Fontana \& Buss introduced AlChemy to resolve what they called the ``existence problem of population genetics''~\cite{fontana1994arrival, fontana1994would}. This problem is related to the origin of life, and other major transitions in the organization of living material~\cite{szathmary1995major}. By {\it organization}, we mean a system with interacting components, which is stable through time, and that stability is not due to the material stability of any individual component but instead is generated dynamically by the relationship between itself and the other component parts of the system. Examples of organizations include living cells, autocatalytic chemical reaction networks, as well as ecosystems, economic firms and the biosphere itself.  It is difficult to determine which features of a stable organization are necessary and which are contingent. To claim that a feature is contingent requires demonstrating an alternative form of the organization without that feature---or at least with the feature instantiated by different means~\cite{fontana1994would}. In the context of the origin of life, this is a problem because we can characterize the dynamics of (bio)chemical reaction networks as we know them, but these dynamics themselves are incomplete explanations for the existence of the underlying network. We are less interested in ``what are the dynamical properties of this reaction network?'' than we are in ``why does biology use these reactions, instead of the vast ensemble of alternatives?~\cite{smith2004universality}'' AlChemy was designed to address the latter question, by providing a model to study the emergence of many stable organizations and compare them~\cite{fontana1994would}.

AlChemy abstracts away the details of real chemistry and focuses on three key features: (i) a vast combinatorial space of objects that can be constructed from a finite set of basic building blocks, (ii) interactions between objects lead to the production of new objects, and (iii) the outcome of interactions is determined completely by the internal structure of the objects involved~\cite{fontana1994arrival}.  These abstract properties were implemented using the formal model of computation known as the $\lambda$ {\it calculus}~\cite{fontana1993beyond, fontana1994arrival, fontana1994would}. In computer science, $\lambda$ calculus played an important role in the development of the theory of computing, and arose around the same time as Turing machines \cite{church1985calculi}. It is also the basis of functional programming languages, most famously, Lisp.

\subsection{The $\lambda$ calculus}
We next give an informal description of $\lambda$-calculus.  A more formally-inclined reader may wish to consult a rigorous  treatment, e.g.,~\cite{barendregt1984lambda}, and the uninterested reader may decide to proceed knowing only that the $\lambda$-calculus specifies a set syntactic expressions, rules for combining and transforming them (as in molecular reactions), and rules for simplifying expressions.

In $\lambda$ calculus, objects are defined as $\lambda$ {\it expressions}.  A $\lambda$ expression takes one of the following forms:
\begin{enumerate}
    \item A single variable, $x$, chosen from some finite set of symbols, $\Sigma$.
    \item A lambda abstraction, $\lambda x\ .\ E$, where $x$ is a variable and $E$ is an expression. This form describes a function that binds $x$ as an argument to its body, $E$.
    %$\lambda$ can be thought of as a binding operator, which binds variables in expressions.
    \item An application $(E_1)\ E_2$, where both $E_1$ and $E_2$ are expressions. This form describes the composition of objects. If $E_1$ is a function, then $E_2$ is the argument to that function.
\end{enumerate}
Intuitively, we have a system of variables, function definitions (called abstractions), and function applications (in which arguments are bound to variables or other functions). However, expressions can contain both \textit{bound} and \textit{free} variables. A bound variable is simply one that is associated with a $\lambda$ abstraction, and all other variables are free. For example, in the expression $\lambda x. ~x ~ y$, $x$ is bound, and $y$ is free. 

In addition to the three basic forms described above, $\lambda$ calculus has two substitution rules for simplifying expressions as far as possible, referred to as a {\it normal form}. These rules are the `calculus' component of {\it $\lambda$ calculus} and are known as $\beta$-reduction and $\alpha$-substitution. The application of these rules to a $\lambda$ expression is conceptually similar to simplifying a mathematical equation.

An expression is $\beta$-reducible if it has the form 
$(\lambda x.~E_1)~E_2$. In such an expression, $\beta$-reduction first substitutes $E_2$ for each $x$ in $E_1$ that is bound to $\lambda x$, and then drops the $\lambda x$.  $\alpha$-substitution uniformly renames a variable in a $\lambda$ expression, similar to renaming a variable in an algebraic expression. When we perform an $\alpha$-substitution of a binding variable $x$ in an expression $E$ by another variable $y$, we substitute the binding $\lambda x$ with $\lambda y$ and all instances of $x$ bound by $\lambda x$ in $E$ by $y$. If $x$ is free, we substitute only the free $x$ by $y$.

We consider some simple examples to illustrate the basics of $\lambda$ calculus, particularly those most relevant to AlChemy. Consider two $\lambda$ expressions, (i) $\lambda x.~x~w$, and (ii) $\lambda y.~\lambda z.~y$. The composition (application) of (i) and (ii) gives the expression (iii) $(\lambda x.~x~w)~\lambda y.~\lambda z.~y$. This expression is of the form $(\lambda x.~E_1)~E_2$, where $E_1 = x~w$ and $E_2 = \lambda y.~\lambda z.~y$. We can use one $\beta$-reduction here, substituting for $x$ by $E_2$ in $E_1$ and dropping the $\lambda x$. This gives $E_2~w = (\lambda y.~\lambda z.~y)~w$.

There is a well-defined computational procedure for generating the normal form of a $\lambda$ expression, if one exists. Isolated variables (expressions with top-level abstractions), and applications without an abstraction on the left-hand side do not admit $\beta$-reduction. Such expressions are already in normal form. AlChemy requires that every expression be in normal form, and so these expressions tend to make up much of the population. However, not every expression has a normal form, because certain expressions may never reduce to a terminating state (this follows from $\lambda$-calculus being Turing complete). An example of such an expression is the $\Omega$ combinator: $(\lambda x.~x~x)(\lambda x.~x~x)$, which reduces to itself in one $\beta$-reduction. We can see this by applying a $\beta$-reduction, observing that this expression takes the form $(\lambda x.~E_1)~E_2$ with $E_1 = x~x$ and $E_2 = (\lambda x.~x~x)$. The substitution of each $x$ in $E_1$ with $E_2$ yields $E_2 ~ E_2 = (\lambda x.~x~x)(\lambda x.~x~x)$. Additional $\beta$-reductions produce the same expression, and thus this expression has no normal form. The $\Omega$ combinator is one of many non-terminating expressions. Worse yet, for any given expression it is computationally undecidable to determine if its reduction will terminate in normal form.  AlChemy takes a pragmatic approach to this problem by attempting to reduce an expression a given number of times before giving up.

Correctly reducing a $\lambda$ expression to normal form with the same meaning may also require $\alpha$-substitution to avoid incorrectly overloading (``capturing'') variable labels. For example, consider the expressions (i) $\lambda x.~\lambda y.~x$ and (ii) $\lambda x.~x~y$. Composing (i) and (ii) gives (iii) $(\lambda x.~\lambda y.~x~y)\lambda x.~x~y$. This expression is reducible: it takes the form $(\lambda x.~E_1)~E_2$, with $E_1 = \lambda y.~x~y$ and $E_2 = \lambda x.~x~y$. If we attempt to $\beta$-reduce this expression directly, we get $\lambda y.~(\lambda x.~x~y)~y$. However, the isolated and free $y$ in $E_2$ is now bound incorrectly (captured) by the $\lambda y$ in $E_1$.  That is, the meaning of the argument $E_2$ is semantically different {\it within} the expression after the incorrect reduction. The meaning of $E_2$ within $E_1$ can be preserved by renaming the conflicting $\lambda y$ in $E_1$ to any fresh (unused) name before reduction, here we choose $z$. The expression $E_1$ is $\alpha$-equivalent to $E_1' = \lambda z.~x~z$, and the composition of $E_1'$ and $E_2$ yields $\lambda z.~(\lambda x.~x~y)~z$. This preserves the free $y$, and does not change the meaning of the argument within $E_1$ after reduction.

\subsection{$\lambda$ expressions in AlChemy}
Using $\lambda$ calculus as the basis of the model, Fontana \& Buss generated what could be described as a ``Turing Gas''--- a collection of random expressions that ``collide,'' where a collision causes one expression to be applied to another, and the resulting expression is reduced to normal form (e.g. $A + B \to A + B + C $, where $C = (A)B $) ~\cite{fontana1994arrival, kauffman1993origins}. This constructive process is interpreted as a catalytic reaction, where every object in the system can serve as a reactant in some reactions, and as a catalyst in others~\cite{szathmary1995classification}. Notice that the composition of $\lambda$ expressions is not commutative ($(A)B \neq (B)A $), which means that for each pair of $\lambda$ expressions, there are two possible reactions, which are chosen with equal probability in AlChemy. AlChemy does not rely on or suggest an inherent ``computational'' nature of chemical systems. Rather, the relevant features of $\lambda$ calculus are (i) an infinite set of possible objects (expressions) that can be generated from a small set of building blocks (variables), and (ii) the ability to generate new objects according to simple interactions between existing objects~\cite{fontana1994arrival}. 

AlChemy simulations proceed as follows: i) initialize the system with a set of $N$ random $\lambda$ expressions, (ii) perform a collision, by picking two expressions at random, applying the first to the second and reducing to normal form, (iii) add the newly generated expression to the system, (iv) remove a random expression to keep the total number of expressions constant, and (v) repeat steps ii--iv until $T$ collisions have been performed.  Time is measured in units of collisions, i.e., each collision corresponds to one time-step. Because $\lambda$ expressions are selected according to their abundance in the soup, this process obeys the principle of mass action, meaning that the rate of reactions is proportional to the product of the reactant concentrations. \textcolor{black}{ The model is of interest because when one starts with distinct, arbitrary, and randomly generated expressions, the procedure does not produce more distinct, arbitrary, or otherwise random expressions. Instead, the simulations generate a collection of $\lambda$ expressions which collectively reproduce each other, which have interrelationships that are not predictable by the rules of $\lambda$ calculus alone, analogous to the relationship between biochemistry and standard physics~\cite{fontana1994would, sharma2023assembly}}.

To implement the model, two additional features are required: termination and filtering. The reduction of an expression is not guaranteed to terminate, and there is no way to know ahead of time which ones these are (a manifestation of the famous Halting problem). Therefore, in AlChemy, when two expressions collide, it is possible that the resulting product may never reduce to normal form. This infinite regress is prevented using {\it pragmatic reduction}. A finite limit is placed on how many times each expression is reduced. If the reduction does not terminate by the limit, the reaction is deemed {\it elastic}, and the original expressions are returned to the simulation. Finally, to explore different boundary conditions on AlChemy's dynamics, the original authors included an option of excluding certain reactions based on pattern matching, referred to as {\it syntactic filters}. For example, the original work includes simulations in which copy actions are excluded by identifying reactions that produce an explicit copy action (of the form A + B $\rightarrow$ 2A + B ) by comparing the products and reactants and labeling any such reaction as ``elastic.'' Syntactic filtering allows AlChemy dynamics to be modified by removing entire classes of reactions.  The impact of such filters has not been studied carefully, except the case just mentioned where copy actions are eliminated. Simulations can be run with our without syntactic filters imposed.

\section{Prior Work}
 
The original Fontana \& Buss paper was a landmark study in constructive dynamical systems and inspired subsequent work across several domains~\cite{dittrich2001artificial, quinkert1996aspects, cederman2002endogenizing}. In the field of artificial life some authors proposed modified versions of AlChemy based on combinator logic (rather than simple $\lambda$ calculus)~\cite{di2000less, virgoopen, kruszewski2022emergence}, producing systems admitting reversible reactions or constraints such as conservation laws analogous to conservation of mass.  Those simulations found results that are qualitatively similar to those originally reported by Fontana \& Buss~\cite{di2000less}, though with additional rules imposed on the system. We did not pursue combinator logic here, in part because we were interested in exploring and reevaluating the original AlChemy simulations with the original code base. One of the most interesting aspects of using $\lambda$ calculus as a constructive model is that it allows for open-ended evolutionary dynamics (up to some practical computational limit). Inspired by this, Masumoto \& Ikegami used $\lambda$ calculus and genetic programming to explore the evolution of strategies in game theory, modeling agent strategies and the game itself as $\lambda$ expressions~\cite{masumoto2001lambda}.  Similarly, rule-based modeling approaches have been adopted in the Kappa language to more closely mimic biochemical systems and gene regulatory systems~\cite{boutillier2018kappa}.

%% Review results of Fontana \& Buss
\subsection{Results from the original AlChemy papers}
In their original work Fontana \& Buss define a hierarchy of functional organizations they observed in AlChemy~\cite{fontana1994would, fontana1994arrival}, identifying three distinct organizational levels, known as `L0', `L1' and `L2' organizations. The simplest, L0, is described by the authors as the typical fixed point of a simulation. It is characterized by the dominance of simple copy functions which emerge easily and quickly take over the system. This occurs because the function of copying any input is trivial in $\lambda$ calculus, the expression $\lambda x_1. x_1$, will copy any function it is applied to, and thus highly likely to emerge in most random samples of $\lambda$ expressions. It is easy to detect L0 organizations in a simulation because when they emerge the number of unique expressions in the system tends (down) toward one (because copy actions make more of the copier which dilutes out other species). 

The next level of the hierarchy is L1 organizations. The original experiments generated L1 organizations when syntactic filters were added to the system to prevent all simple copy actions, i.e., those that characterize L0 organization.  With this added constraint, L1 organizations emerged.  They are characterized by autocatalytic sets of expressions, in which each member of the set can be produced by the interaction of other members of the set, even though no member reproduces (copies) itself directly~\cite{kauffman1993origins}. L1 organizations can have much richer dynamical properties than L0 organizations. By definition, the number of unique expressions contained in the organization must be greater than 1, a clear signature that differentiates them from L0 organizations. L1 organizations are robust to perturbation; in fact they are identified in the original work by running the simulation, and repeatedly perturbing the system through the addition of random expressions. In addition, L1 organizations are not necessarily closed under interaction; at any given time in an L1 organization the composition of two expressions could yield a new expression not currently in the system. Over time, however, these new expressions will be diluted out of the system because they are not being generated consistently. Thus, the long term evolution of the system of expressions tends to converge to some stable distribution with finite fluctuations around it. Given the size and diversity of the combinatorial space of expressions, there is no guarantee that all (or even most) L1 organizations will converge to the same distribution. In the original work, the authors give examples of several distinct L1 organizations, each with its own unique internal structure and logic. The rules of these systems are consistent with, but not explained by, the rules of $\lambda$ calculus itself, just as the apparent rules governing living systems are consistent with but unexplained by the laws of physics and chemistry as we know them. 

The highest level of organization characterized in the original work are `L2' organizations. These are composites of L1 organizations and can be discovered either spontaneously using the same constraints as the L1 organizations, or by manually composing two previously discovered L1 systems. L2 organizations are characterized by the existence of two or more L1 organizations and additional expressions (called `glue' in the original work). The `glue' expressions could not exist without at least one of the L1 organizations and are produced by composing functions from different organizations. L2 organizations may be difficult to identify when they emerge spontaneously because deciding whether an organization can be split into two distinct stable organizations is a difficult problem to solve without resorting to trial and error.

To avoid confusion we distinguish between ``organizations'' and ``simulations.'' When we say an `L1 simulation'' we mean the simulation parameters described by the original authors that generated ``L1 organizations.'' So, for example, while our analysis shows that ``L0 simulations'' can produce complex organizations, we will continue to refer to them as ``L0 simulations,'' without commenting on where the results belong in the organizational hierarchy.

\section{Experimental Results} 

This section first details how we recompiled the original code in a modern computing environment, and then reports the results of our investigation into the organizational hierarchy described by Fontana \& Buss. We were largely able to reproduce their key simulation results, with some interesting elaborations on the statistics of the originally reported behaviors. This was possible only because the original authors archived their implementation and hosted it publicly. This provided a unique opportunity to revisit a landmark study nearly three decades later with the advantage of modern computing resources. 

Much of the original AlChemy project was written in an early version of C, and the code has several incompatibilities with modern compilers.  Therefore, it was not possible to compile the exact original source with modern compilers such as \verb|gcc 4xx|. Accordingly, we made a small number of minor modifications, which to our knowledge do not change the semantics or behavior of the original program.

The biggest change was that the code could not be compiled with the version of \verb|camllight| distributed with the code.  Instead, we adopted a slightly different version of \verb|camllight|~\cite{dougs_camllight}. Beyond this only minor modifications were required, which primarily involved changing \verb|include| statements, and making minor changes to function names. Specifically we had to to include the headers \verb|stdlib.h|,  \verb|string.h|,  \verb|time.h|,  in several files. We renamed one function in \verb|LambdaReactor/main.c| originally called \verb|select(...)| to \verb|peep(...)| which avoids a naming conflict. This function is called only once in that file and nowhere else. In one function (\verb|Original_reaction()| in \verb|LambdaReactor/interact.c| we updated the memory allocation call to use \verb|calloc()|, rather than a function in the original code base, \verb|space()|, which was defined in \verb|LambdaReactor/utilities.c|. With these changes, the original source compiled successfully in our environment and are available at ~\cite{our_github}, which includes a docker container, and Python scripts to run the code and analyze outputs.

The dynamics of AlChemy are easy to infer based only on a description of the simulation.  As an example, Figure~\ref{fig:example} shows an organization that emerged in a simulation using the recompiled code.  The simulation was initialized with 1000 random expressions generated with a maximum depth of 7 (see Section~\ref{sec:dynamical-consequences}), pragmatic reduction was set to a maximum of 500 reduction steps, and copy actions were prohibited. The simulation ran for a total of 500k collisions. This simulation was part of our statistical survey, selected after the fact as a useful example because it contained only four expressions at the end of the run. It illustrates the simple L1 organizations identified by Fontana and Buss and described in Figure 2 of \cite{fontana1994arrival}. Figure \ref{fig:example}A shows the actual $\lambda$ expressions, which we label with $a,b,c, \& d$. Figure \ref{fig:example}B, shows the time-series of the organization. The horizontal axis indicates the collision number of the simulation (time), and the vertical axis shows the relative abundance of each unique expression (the vertical placement of each expression is arbitrary---only its relative area on the plot is relevant. The simulation is initialized with 1000 random expressions, and each unique expression is assigned a gray-scale color, except for the four expressions in the final stable organization (i.e., $a,b,c \& d$). The reaction rules implied by the expressions are shown in Figure \ref{fig:example}, and its network representation is shown in \ref{fig:example}D. There are only 12 reaction rules, instead of the 16 possible combinations of four expressions, because four combinations would  produce copy actions and be filtered.
This example shows how AlChemy simulations can produce random expressions that repeatedly interact to produce a stable, self-consistent organization with its own internal logic.

\begin{figure}[h!]
\centering
\centerline{\includegraphics[width=\textwidth]{"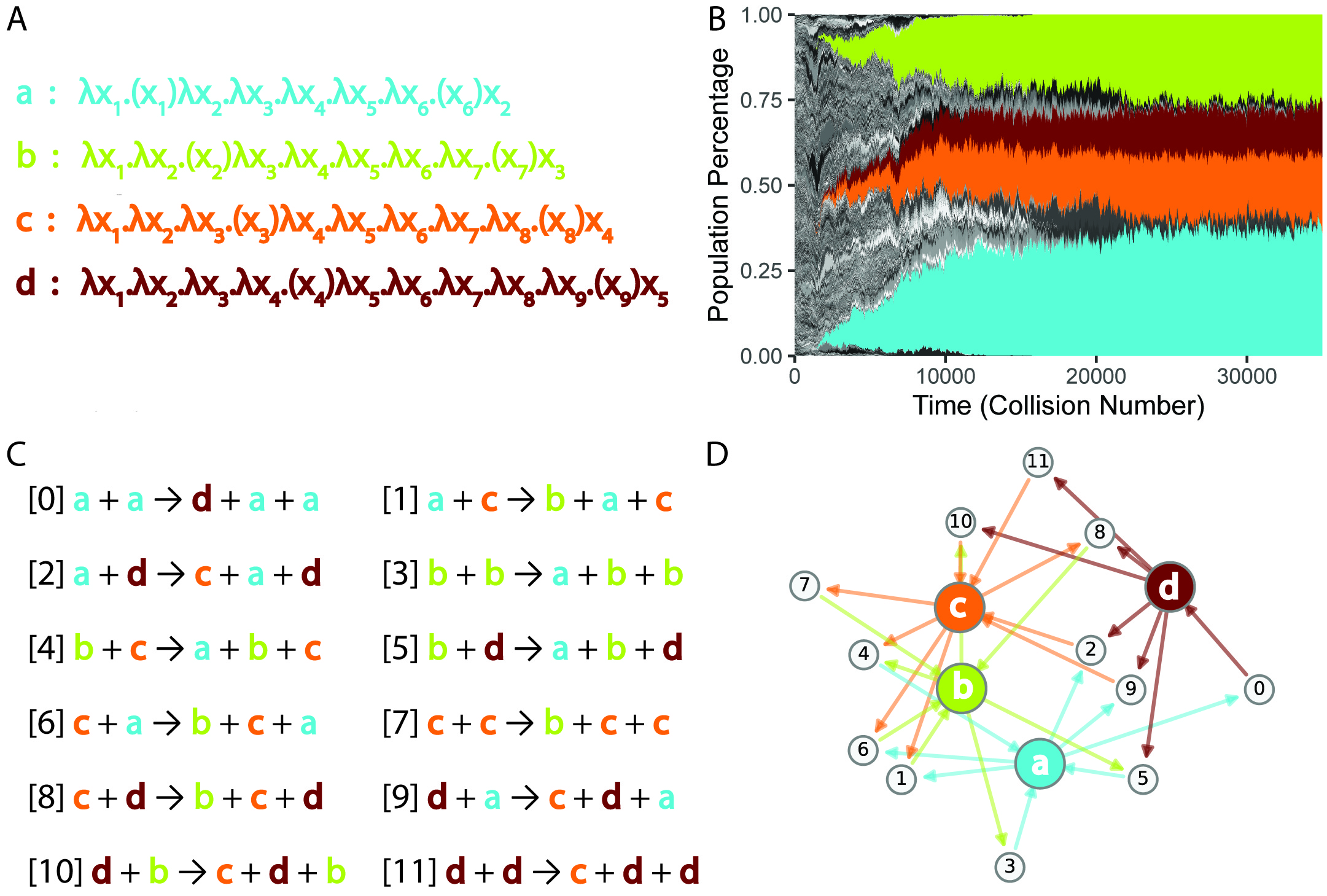"}}
\caption{Example AlChemy simulation and the organization it produced. (A) The four $\lambda$ expressions remaining at the end of a single simulation run, which constitute the L1 \textit{organization}. (B) Timeseries of the simulated run where the vertical axis represents the fraction of the population occupied by any given expression. Each expression is assigned a grey-scale value, except the four winners shown in (A). After a few hundred collisions all four expressions in the final organization had emerged, and by $\approx$ 35k collisions they were the only expressions remaining in the simulation. (C) The reaction rules that correspond to the final organization, showing that it is closed under interaction between any two expressions. (D) Network representation of the reaction rules from (C).}
\label{fig:example}
\end{figure}

\subsection{L0 Organizations}

\begin{figure}[h!]
\centering
\centerline{\includegraphics[width=\textwidth]{"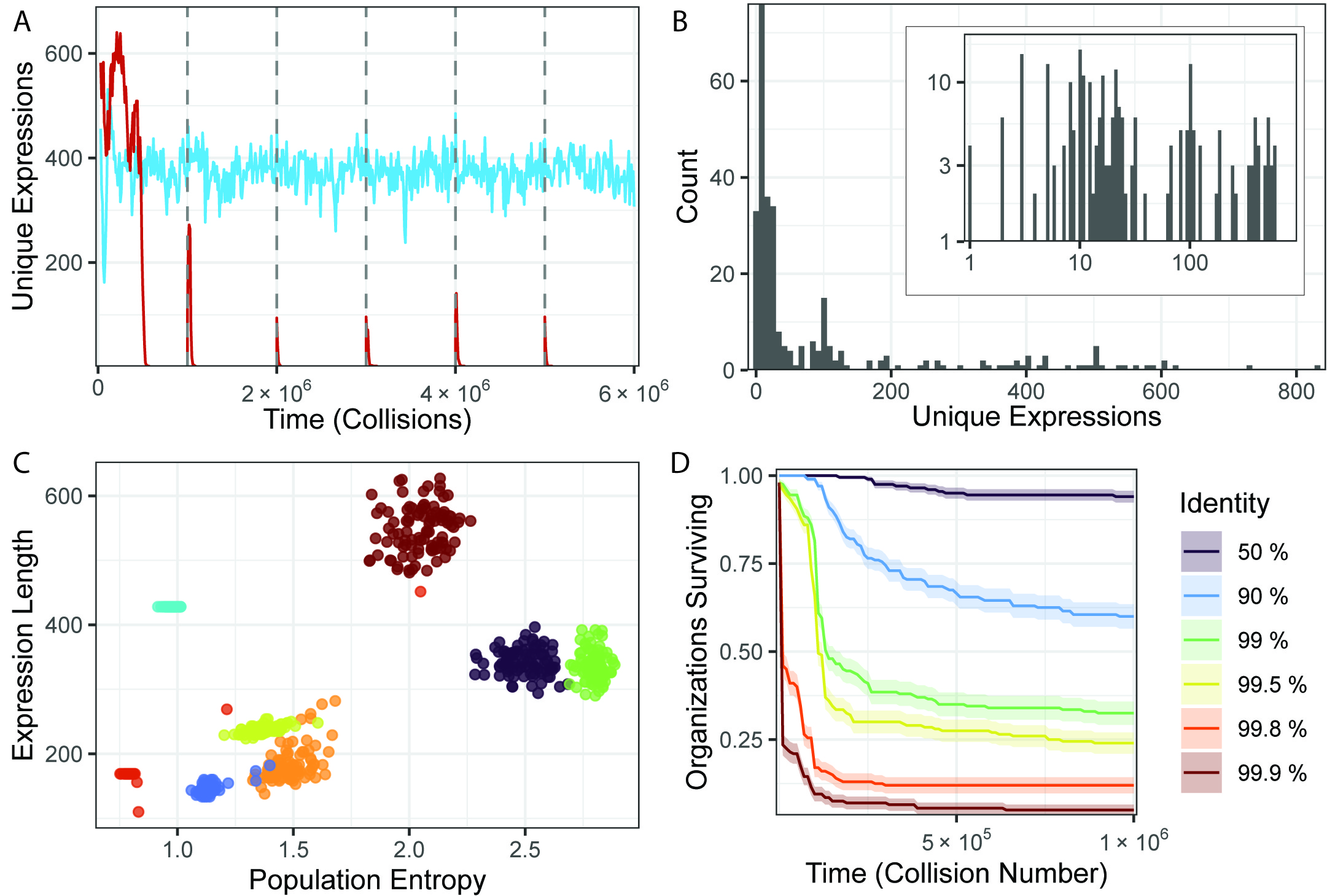"}}
\caption{L0 Simulations. (A) Time series from two different representative simulation runs (red and blue), showing the number of unique expressions through the runs. In both cases the maximum number of objects was set to 1000 and a total of $6\times10^6$ collisions were performed. Every $10^6$ collisions (grey dashed lines) we introduced a perturbation by adding 100 random expressions to the simulation. (B) The distribution of unique expressions across 1000 different simulations after the fifth perturbation (median across the previous $10^6$ timesteps of the simulation), inset shows the same data on log-log scales. The number of expressions at steady state seems to vary across orders of magnitude and is heterogeneously distributed. (C) The organizations stabilize with fixed macroscopic parameters, even though the expressions themselves are continually changing, the average expression length and population entropy (a measure of diversity) are shown for different simulations. Each color is a different simulation with the identical run parameters, each point represents a snapshot of the simulation at different time points. (D) Survival Rate over time of different organizations when they are perturbed by replacing $p$\% of the expressions with the identity function. Different colors correspond to different values of $p$.}
\label{fig:L0sims}
\end{figure}

Recall that L0 organizations are expected to emerge in the least constrained version of \textit{AlChemy}, as they are dominated by trivial copy actions (e.g., the identity function $\lambda x_1.x_1$, applied to itself). Fontana \& Buss detected these by running the simulation without syntactic filters, and iteratively perturbing the system by adding new random expressions after the simulation ran for a long time. When the number of unique expressions converged to one or just a few expressions, the simulation invariably contained identity functions that copy themselves through application to other identity functions. This trivial fixed point (e.g. all $\lambda x_1.x_1$) exemplifies the L0 organizations. We recreated this experiment using wrapper scripts that call the original code and manipulate the input/output files (available online~\cite{our_github}). Specifically we simulated 1000 unique expressions for over 1,000,000 collisions, then we added 100 random expressions to the system and ran it for another 1,000,000 collisions, repeating this process a total of five times. See section \ref{sec:random-exprs} for details on how random expressions are generated. Two example time-series from this experiment are shown in Figure \ref{fig:L0sims}(A), one of which illustrates the expected behavior (shown in red). Here the dashed lines indicate the times at which perturbations were introduced. The two colors represent two different runs of the simulation with different random seeds. 

In contrast with the original results, we found some L0 simulations that do not end in trivial fixed points, dominated by copy functions. In some of our simulations more complex organizations emerged, which are characterized by a large number (10--100s) of distinct expressions and robustness to repeated perturbations. These organizations are consistent with the description of L1 organizations and possibly even L2 organizations. The blue line in Figure \ref{fig:L0sims} (A) shows an example. This simulation had a steady state of about 380 unique expressions, and it was unaffected by perturbations. The distribution of steady-state unique expression counts is shown in Figure \ref{fig:L0sims} (B). Many simulations end with only a few species, but some simulations end with 10s or 100s of unique expressions. This was unexpected given that these systems should not be robust to the addition of copy expressions. These results are interesting for the same reason the `L1' organizations in the original work are interesting: they demonstrate that stable organizations can emerge with complex internal structure, from seemingly random inputs and minimal rule sets. The only constraints imposed on this system are the composition rules of $\lambda$ calculus and an initial set of expressions. From these the simulations produce an organization defined by a network of mutually dependent interactions, which were not manually encoded in the initial conditions. This type of self-organized complexity is exactly the type of phenomena we would like to be able to observe in chemical systems~\cite{surman2019environmental, cronin2016beyond}. 

As in the original work, we detect stable organizations that contain enormous diversity. We can characterize this diversity by measuring average properties of an entire population and comparing them against each other, for example the average length of expressions and the population entropy. By population entropy, we mean the diversity of the expressions in a system, which we can calculate by making a species distribution curve and calculating its entropy. Figure \ref{fig:L0sims}(C) shows some of this diversity. Each color in this panel indicates a different simulation run with the same parameters, and each point is a single snapshot (point in time) in the simulation (after all perturbations have been performed, e.g., the previous 100 snapshots). The horizontal axis shows the population entropy (diversity) of the simulation at the given snapshot. If each of the 1000 expressions were unique, the entropy would be 3, and if it contained only copies of a single expression, the entropy would be 0. The vertical axis shows the average length of the expressions in the simulation at that time. These two parameters do not mutually constrain each other in the dynamics. However it is interesting that once an organization has been established, its own internal dynamics determine the diversity of the population, the length of the expressions, and the variation of those parameters through time. For example, the orange and blue points show similar values of population entropy and expression length, but the blue simulation has less variation than the green. Some stable organizations are tightly constrained by their internal dynamics, while others enable wide variation in their averaged  properties. 

Given the surprising emergence of larger stable organizations, we wondered if they were in fact robust to copy functions, or if copy functions had never emerged in the system, and when introduced if they would lead to the destruction of the stable organizations. To test this we initialized 200 simulations with the $\lambda$ expressions from the end state of 50 different simulations (e.g. a random subset of those shown in Figure \ref{fig:L0sims}B). We used each such end state to generate four initial populations for further simulation. In each initial population we replaced a random $p$ percentage of the expressions with the identity function $\lambda x_1. x_1$, and then ran the simulation for an additional $10^6$ collisions. At each step we recorded whether the organization had ``survived,'' by evaluating whether there were any expressions other than the identity expression remaining in the system (in this case `dying' means collapse into the state of only $\lambda x_1.x_1$). Figure \ref{fig:L0sims}(D) shows the results of this experiment, for different values of $p$. Surprisingly, the organizations required very large perturbations to be destroyed. Even replacing 50\% of the expressions with the identity function did cause collapse, and 90\% replacement destroyed only some, but not most, of the organizations. Even replacing 99.9\% of the expressions (e.g. leaving only a {\it single} expression other than the identity), led to relatively long transients before the organizations collapsed. 

One reason for this surprising robustness is the non-commutative nature of AlChemy. The application of $\lambda x.x$ to any function will return the function itself, but the application of any function $f$ to $\lambda x.x$ will not necessarily return $f$ or $\lambda x.x$. This feature, combined with the fact that all reaction are ``catalytic'' ($A + B \rightarrow A + B + C$) enables relatively long transients, even when the majority of interactions are copy actions. Importantly, when relatively large perturbations are applied to organizations (as in Figure \ref{fig:L0sims}D), the organization is often fundamentally altered, bearing little resemblance to the organization before the perturbation. These results demonstrate that the most likely fixed point (only $\lambda x.x$) is difficult to access using expressions produced by a long AlChemy simulation.   

\subsection{L1 Organizations}

\begin{figure}[h!]
\centering
\centerline{\includegraphics[width=\textwidth]{"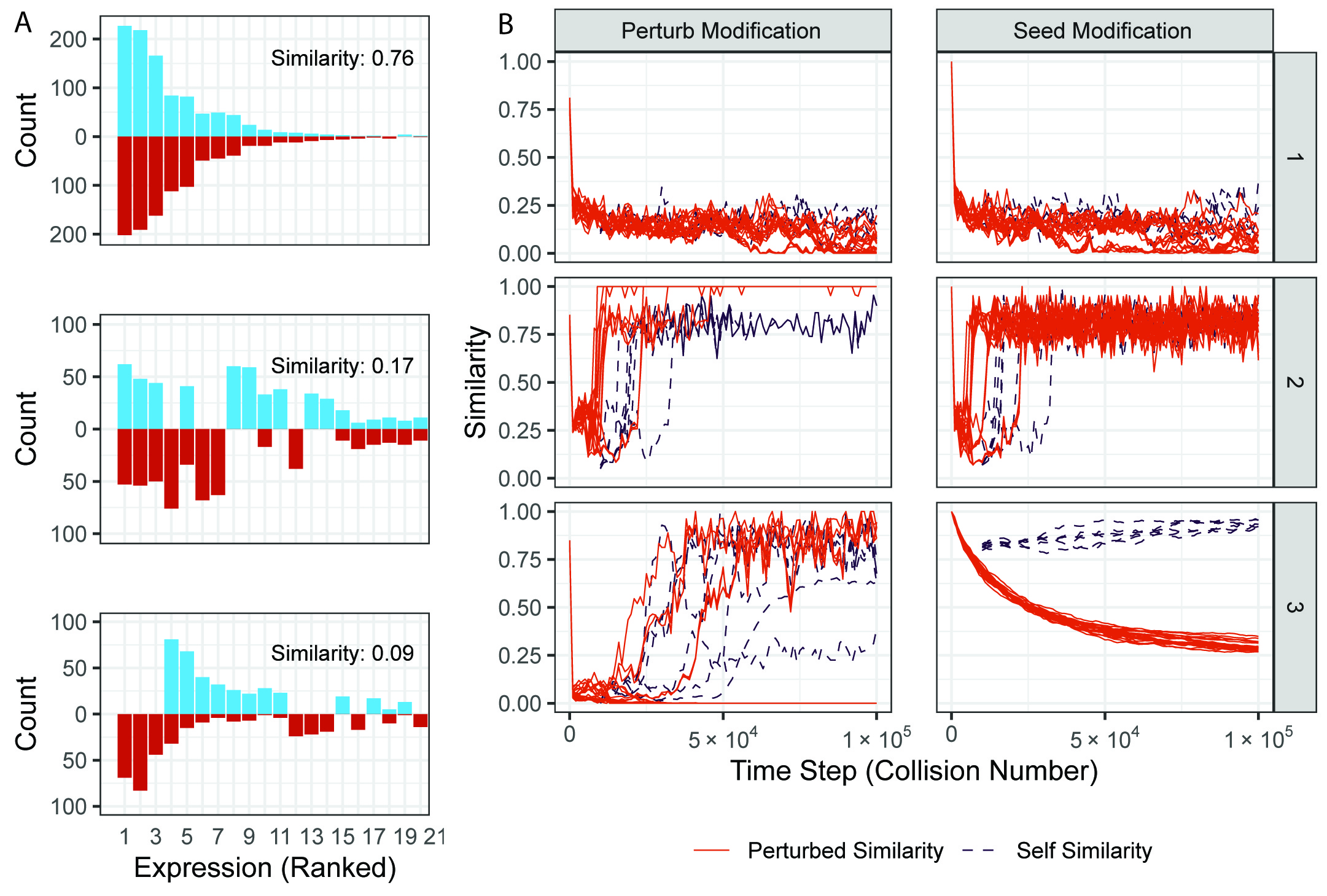"}}
\caption{L1 Simulations and Similarity over time. (A) Abundance Distributions illustrating pairs of similar organizations. Each plot shows two different pairs (Blue and Red); each bar corresponds to one expression that is a component of each member of the pair; the ordering of the expressions is determined by the sum of their abundance in both members of the pair. The vertical axis indicates the expression copy number with the red values inverted to facilitate comparison. High symmetry between blue and red indicates similar copy number, low symmetry indicates different copy numbers between the organizations. (B) Stability analysis of four L1 organizations that emerged in simulation. Each organization was initialized with a modification, either ``perturb" (10 random expressions added), or ``seed" (different random number seed only). The similarity between the modified and original simulations is shown in solid orange lines. The similarity of the original simulation to \textit{itself} at a previous time-step (500 collisions earlier) is used as a control, and it is shown in the purple dashed lines. Each panel shows a different input organization, selected to highlight the diversity of possible outcomes.}  

\label{fig:L1sims}
\end{figure} 

Fontana \& Buss used syntactic filters to eliminate the effect of copy functions (like $\lambda x.x$) (as described above). Although our analysis shows that these functions are unlikely to dominate the system in general, we followed the procedures outlined in the original work. We ran L1 simulations following the protocol described above for `L0' except that we imposed syntactic filters to bar copy-actions. We refer to the resulting organizations as ``L1 organizations'' even though they appear similar to many organizations we found with L0 simulations. Our simulation results are consistent with those described in the earlier work. Given the robustness of the organizations we observed with L0 simulations (figure \ref{fig:L0sims}), for L1 we focused our attention on statistically analyzing the robustness of discovered organizations to random perturbation, instead of the targeted perturbations shown in Figure \ref{fig:L0sims} D). 

We measure the similarity of $\lambda$ expressions at any given point in time using the Jaccard Index~\cite{levandowsky1971distance} as a measure of similarity between two sets (organizations).  The Jaccard Index is simply the ratio between the intersection and the union of two sets. When the Jaccard index is close to $1.0$ the two sets contain nearly identical expressions, when it is close to $0.0$ they are nearly disjoint. We note a potential weakness of the Jaccard Index is that it does not account for the relative abundance of expressions in a population. Thus, in the following, whenever we refer to the similarity between $A$ and $B$, we mean the Jaccard Index of $A$ and $B$.  

We consider the copy frequency (count) in Figure \ref{fig:L1sims}(A).  We ask, for two organizations that have a given similarity at the same time step, how are the expression counts distributed (i.e., how many copies are there of each different expression)? Figure \ref{fig:L1sims}(A), plots the distribution for three different examples, each illustrating a different level of similarity. Each plot shows two different organizations (Blue and Red) each bar indicates a different expression, and the size of the bar indicates the copy number of that expression in the organization (red is inverted to show a clear comparison to blue). The bars are sorted according to the sum of the copy number in the red and blue organization. When similarity is high, the expressions in both organizations are nearly the same, and we find that corresponds to similar copy numbers for each expression (even though the Jaccard index does not track this explicitly). When similarity is low, the number of expressions common to both organizations is low, and the copy number of the shared expressions are often different. This suggests that when stable organizations emerge, their relational structure is determined by the expressions they contain and copy number is largely determined by that structure. A priori, this need not be true, it is not difficult to construct chemical reaction networks that exhibit bi-stability, which would mean a single set of compounds and reactions could generate two distinct sets of concentrations. But, this does not appear to be the case for the organizations found here. For organizations discovered by the L1 simulations, bi- or multi-stability appears rare. Although bi-stable and oscillatory CRNs are ubiquitous in biochemical reactions, and trivial to produce theoretically, empirical systems exhibiting bi-stability are relatively rare, and often must be engineered into a system through fine control \cite{maity2019chemically, semenov2016autocatalytic}. 

We next considered the stability of the organizations through time and how they respond to perturbation. We measured the similarity of a simulation state to itself at an earlier time, comparing the state at time step $t$, and $t-n$, for several different values for $n$. We found different but similar trends for values of $n$ between $100$ and $2000$ collisions, and we show results for $n=500$ in Figure \ref{fig:L1sims}B (dashed purple lines). In general, even when the number of unique expressions in the simulation remains stable over time, the internal structure of the organization is not as stable---the expressions are continually changing over time, to varying degrees.

To study this turnover of expressions, we investigated both perturbations to the expressions themselves and changes to the simulation procedure (how collisions are selected).  We first make small changes to the state of the system (e.g., $x \rightarrow x + \epsilon$), which involved removing 10 randomly selected expressions, and adding 10 randomly generated expressions. In deterministic dynamical systems this is what many practitioners mean by perturbation~\cite{ellner1995chaos}, and we refer to this as ``perturb.'' Second, we consider a modification that changes the order in which expressions collide~\cite{ellner1995chaos}, referred to as ``seed.'' In AlChemy, collision order is determined by the pseudo-random number generator, and in this modification we ran the simulation with a different integer as the seed for the pseudo-random number generator.

We randomly selected different L1 organizations from the simulations and performed each modification independently, which resulted in two new organizations that we ran for $10^5$ additional collisions. For each resulting simulation we measured the similarity between the modified simulation and the unmodified simulation, as well as between the modified simulation and itself $500$ collisions in the past. We repeated this experiment seven times for each organization. We observed a diversity of possible outcomes, we show three representative organizations in Figure \ref{fig:L1sims} (B). Each panel shows the similarity of a modified system to its unmodified counterpart, two panels for each of the three organizations, one for each type of modification.

In organization 1 both types of modifications often led to different organizations, sometimes not overlapping at all with the original, as does the self-similarity control, suggesting that the original system was not truly stable. In organization 2 both types of modifications eventually produce an identical or nearly identical organization to the original, as does the self-similarity control. Finally, organization 4 demonstrates that a single organization can produce different responses depending on the type of modification. In some cases (but not all) ``perturb'' led to a new organization, while in others the organization eventually returned to its original state (e.g. similarity $\sim$ 1). In all cases, however, the ``seed" modification caused systems to slowly drift apart into distinct states (low similarity), and these states were stable through time (high self-similarity). These results show that the stability of `L1 organizations' varies widely, both through time and in robustness to external perturbations. In some cases (as in organization 2) they are highly robust, while in others (organization 1), they are highly sensitive. These differences suggest that the organizations we found respond differently to environmental changes. This makes them interesting targets for a Darwinian process, because we know, e.g., that different organizations will have different capacities to respond to selective conditions. A promising future direction for this work would include competition and selection among different organizations. 

\subsection{L2 Organizations}
\begin{figure}[h!]
\centering
\centerline{\includegraphics[width=18cm]{"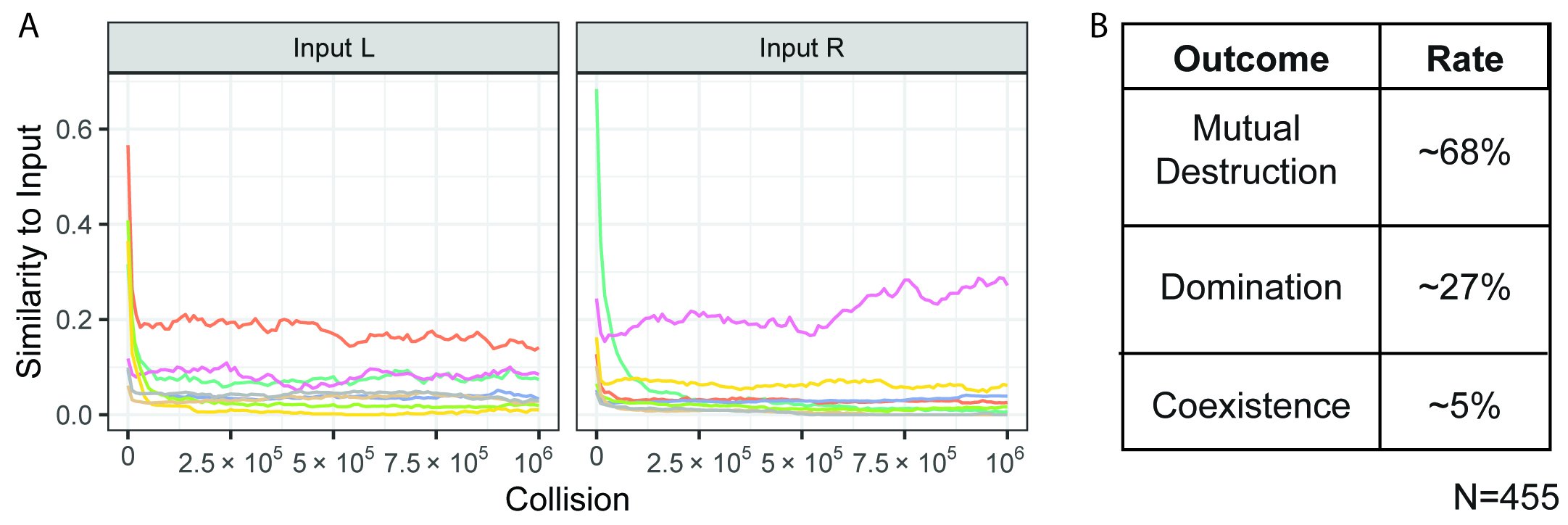"}}
\caption{L2 Simulations. Composite organizations were formed by combining two different L1 organizations ($L$ and $R$) into the same simulation and running the combination for $10^6$ collisions. (A) time-series of the similarity of the composite to the inputs; color corresponds to different choices of $L$ and $R$. (B) three possible outcomes: (i) Mutual Destruction (the composite retains almost no similarity to either input); (ii) Dominance (one input eventually dominates the other, and the final organization reverts to the dominant input); (iii) Coexistence (the final organization retains  significant similarity to both inputs). The relative frequency of these outcomes is summarized in the table.}
\label{fig:L2sims}
\end{figure}

The highest level described by Fontana \& Buss are L2 organizations, which are composites of stable L1 organizations. In the original work, Fontana \& Buss identified L2 organizations in two different ways, first by combining two independent L1 organizations found in different simulations, and then by identifying an L1 simulation in which two separable organizations emerged. We did not attempt the latter. It is possible that the results we reported in the previous sections were in fact L2 organizations with separable internal structure. Instead, we interrogated whether combinations of our identified L1 organizations could coexist. We started with two L1 simulations ($L$ and $R$) each with $1000$ expressions remaining at the end of a series of over $10^6$ collisions and five perturbations, and combined them into a single simulation with the maximum number of objects set to $2500$.  We then ran that simulation for $10^6$ more collisions and measured the similarity of the system to the original inputs ($L$ and $R$). Figure \ref{fig:L2sims}(A) shows results for eight representative runs, each of which was initialized with different choices of $L$ and $R$. The left two panels show similarity to the $L$ and $R$ inputs over the simulations, and the colors in each panel correspond to the same simulation. In many cases, one input organization dominated the other, e.g., the orange organization, which retains its similarity to the original $L$ input.  Likewise, it retains $0$ similarity to the $R$ input. We partitioned the outcome of these simulations into three coarse categories: (i) Dominance---the organization retains non-zero similarity to one input but not the other, (ii) Coexistence---the organization retains an average similarity $>0.1$ to both inputs, and (iii) Mutual Destruction---the organization retains $<0.1$ average similarity to both inputs. Figure \ref{fig:L2sims}(B) summarizes the relative occurrence of these outcomes across 455 pairs of simulations. These results show that the organizations produced by AlChemy can possibly coexist, but it rarely occurs for organizations evolved in different simulations. 

\section{Generating Random $\lambda$ Expressions}
An AlChemy run is initialized with random $\lambda$ expressions. The system can be sensitive to its initial conditions, and thus the distribution of random $\lambda$ expressions from which the initial conditions are sampled affects the output.  We describe and study the existing method of generating $\lambda$ expressions, used in the original AlChemy, and a second method that samples expressions more `uniformly.' The original method uses a probabilistic grammar as its core random object, and the second method uses a random binary tree. For AlChemy to produce nontrivial results, free variables must be removed from the generated expressions. This can be done by binding the free variables, a process called {\it standardization}, and the way expressions are standardized can have interesting dynamical consequences for the simulation. We study the consequences of two random expression generators empirically.

Generating {\it random} $\lambda$ expressions is not as simple as generating a random binary variable, or a random floating point variable. The probabilistic grammar approach (original AlChemy) leveraged the properties of $\lambda$ calculus. Imagine reading a $\lambda$ expression from left to right; there are three possible first characters, indicating application, abstraction, or variable. The original generator selects one randomly according to three probabilities $p_1$, $p_2$ and $p_3 = 1 - p_1 - p_2$ respectively. When the system generates an abstraction, the variable bound to that abstraction is set to be distinct from any variable previously bound by an ancestor abstraction. When the system generates a variable (with probability $p_4$) it is bound to a randomly chosen parent abstraction if one exists, and with probability $1 - p_4$, it is set to be a free variable. If application is selected, the system recursively generates two new expressions which can be applied to each other, except $p_1$ and $p_2$ change linearly with the depth of recursion. At a depth $d_{max}$, the probability of variable being selected is forced to 1, which terminates the recursion.  This process guarantees that the syntax tree for the generated expression never exceeds a given a depth.

Since there are many possible ways that expressions can be generated from a grammar, we wondered how much the choice of a random generation method affected the results.  We studied a simple alternative, which generates expressions that are distinct from those of the original Alchemy.  Our generator, which we refer to as the {\it Permutation} generator, relies on the observation that the abstract syntax tree of a $\lambda$ expression forms a binary tree. This can be seen by observing that the grammar rule containing the most non-terminal symbols is the application rule ($E \rightarrow (E).\ E$), with two non-terminal symbols.  This forms a tree with $E$ at the root and two leaves, each consisting of $E$.  The tree is formed by assigning the expression on the left hand side of the rule to the root, and assigning one leaf for each of expression on the right hand side of the rule.  A similar process is used, recursively, for sub-expressions that form subtrees.  Because there are never more than two expressions on the right hand side of a rule (and hence, two children for each node), the structures form a binary tree. Using this observation, we first generate a random binary tree, then assign variables to the abstraction and leaf vertices of the tree. We use a standard method for generating random binary trees~\cite{knott1977numbering}. 

Starting with a target number, $n$, of desired nodes, we first generate a random permutation of the first $n$ integers, then construct a binary search tree using the permutation. A binary search tree is simply a binary tree with integer vertex weights and a total order on the children of each vertex, such that the lesser (left) child always has a smaller weight than the parent, and the greater (right) child has a larger weight.~\cite{knuth1997art}. Once such a tree is constructed, the syntactic structure of the expression is determined, and it remains to randomize the semantic structure. To do this, we assign variables to abstraction (one-child) and leaf (zero-child) vertices. Application (two-child) vertices do not have a variable associated with them. As with the original generator, we uniquely assign variables within each abstraction chain, i.e., we set the variable bound by an abstraction to be distinct from any variable bound by an abstraction ancestor, if an ancestor exists. For leaf vertices, we assign a variable at random from a parent abstraction, unless no parent abstraction exists---if so, we assign it to a free variable.

\label{sec:random-exprs}
\begin{figure}[h!]
\centering
\centerline{\includegraphics[width=\textwidth]{"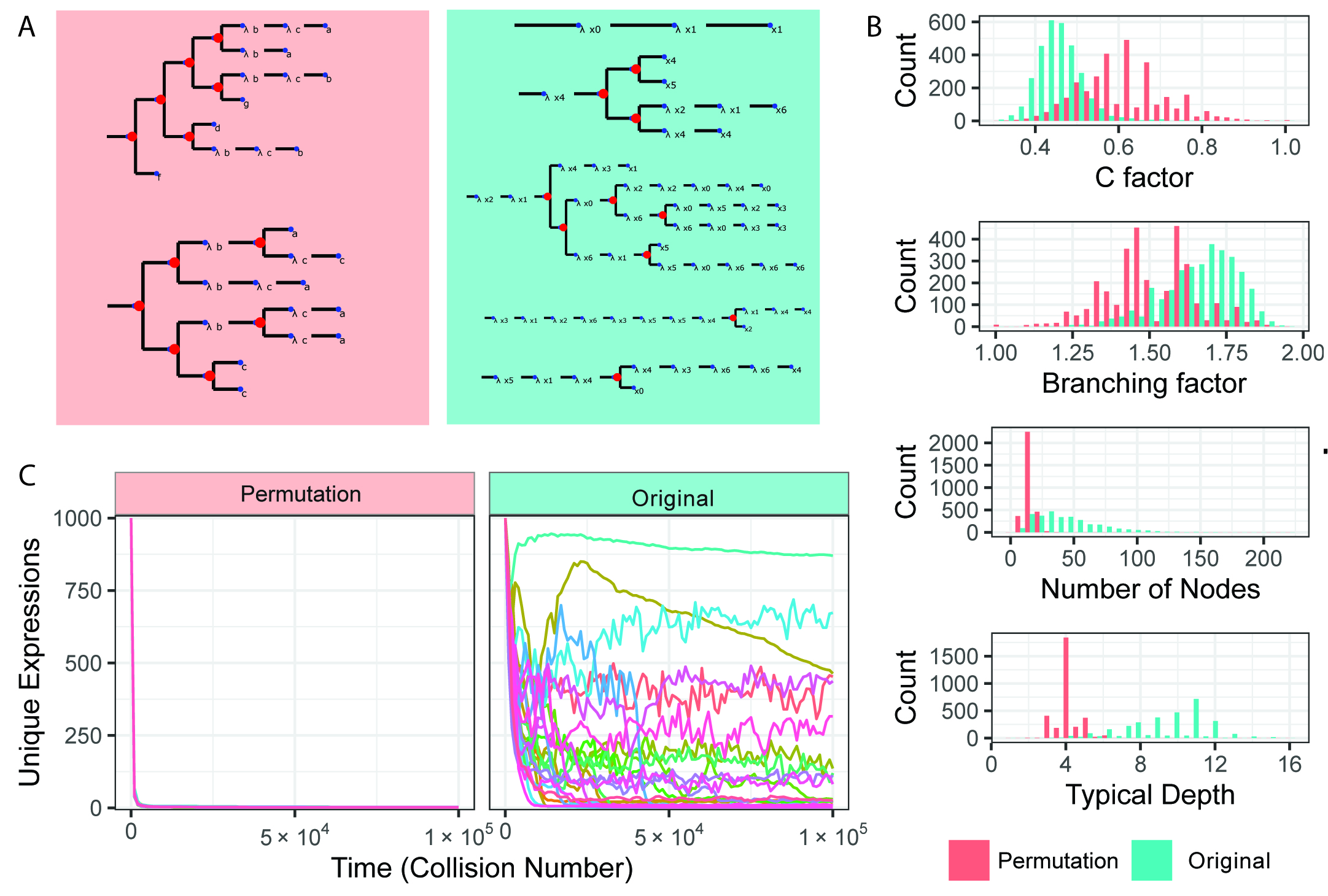"}}
\caption{(A) Representative syntax trees for $\lambda$ expressions generated by the two generators. Left (red): Two example trees generated by the permutation method. These trees are uniform in the ratio of abstractions to applications, and they are parameterized only by the size of the tree.
Right (blue): Five trees generated using Alchemy's original generator. These trees tend to have long chains of abstractions, seen here as long branches. Thus, their corresponding expressions vary greatly in complexity, and simple expressions (such as one or two-abstraction expressions) are common. (B) Statistical properties of the binary tree representation of the $\lambda$ expressions. (C) The dynamical consequences of these different generators are dramatic; in the permutation method (left) the system collapses to an inert, trivial fixed point, while the original generator (right) produces a diversity of complex organizations.  }

\label{fig:randomgenerators}
\end{figure}

\subsection{Standardization}
The process of generating random expressions can produce expressions with free variables. Using the original code, we found that the presence of expressions with free variables in the simulation usually leads to dominance of the term $\lambda x.\ c$, regardless of the syntactic filters. We infer from this that in the original work all free variables were bound (which was an option in the input file for simulation), and we followed that protocol here. The process of removing free variables from expressions is called``standardization," and there are multiple ways to implement it. 

A set of expressions is standardized if each expression in the set is in normal form and has no free variables. The original AlChemy generator standardized expressions as follows: For each free variable $x$ occurring in a generated expression, concatenate the string $\lambda x$ to the head of the expression, thereby binding the free $x$. This is iterated until each free variable in the expression is bound. This method tends to produce expressions with long linear chains emanating from the root of the syntax tree as shown in Figure~\ref{fig:randomgenerators}---in effect, producing expressions that describe functions with a large number of arguments. 

Our permutation generator standardized expressions using a different method, as follows: a free variable can only occur at a leaf vertex of a tree if that leaf has no single-child ancestor (an abstraction node). In this case, because the leaf node is not in the body of any abstraction, the assignment of the free variable is forced. To standardize such vertices, we modify the tree slightly, introducing an abstraction vertex immediately above the free-variable leaf. This binds the child vertex to the newly introduced parent.

\subsection{Dynamical Consequences}
\label{sec:dynamical-consequences}

What is the effect of the different expression generators on the simulation? Figure \ref{fig:randomgenerators}A shows exemplary trees from each generator (permutation in pink, original in blue). Using the binary tree representation of $\lambda$ expressions, we can calculate various properties of the expressions. These are shown in Figure \ref{fig:randomgenerators}B. Simple measures such as the number of nodes in the tree, or the typical depth (median distance between root and leaf) are given. We also calculated the {\it branching factor} which measures the average number of children each node has. Additionally we calculated a variable labeled {\it C factor}, which is the ratio between the maximum depth of the tree and the log of the number of nodes. This measures how ``bushy'' the trees are (if C factor is 1, trees are very bushy/wide, if C factor is 0, the tree is very narrow and stringy). The original generator has a much broader distribution over the number of nodes in the generated expressions, and a broader distribution of typical depths. The original generator, however, tends to produce long stringy expressions with lower C factors and higher branching factors than the permutation method.

To test whether the properties of the original random expression generator had consequences for the experimental results, we ran simulations using the both the original generator and the random permutation generator described above. Using each generator we produced five different random initial conditions, and for each of those conditions we ran the simulations using 5 different pseudo random number seeds. As above, each simulation contained 1000 unique expressions and ran for $10^5$ collision steps. We imposed the `L1' boundary conditions, preventing direct copy actions. The number of unique expressions for these runs is shown in Figure \ref{fig:randomgenerators}C. On the left are the number of unique expressions over time for the initial conditions generated using the permutation generator, while the right shows the results for the original generator. The differences are dramatic. Unsurprisingly, the original generator behaves consistently with the results in Figures \ref{fig:L0sims} and \ref{fig:L1sims}. However for the permutation generator the simulations rapidly collapse into a trivial fixed point containing only the identity function. This is surprising given that the in other case (such as Figure \ref{fig:L0sims}D), this fixed point was stable because of the copying action of the identity function.  However, here we prohibited copy functions, so these identity functions are stable by other means. We verified that when the identity function begins to dominate a run the vast majority of interactions are labeled elastic. So in these cases the identity function is being produced through the interaction of other expressions in the system, and once produced it is essentially inert, because at least 50\% or more of its interactions are elastic collisions. The ubiquitous production of the simplest function from the initial conditions is an unexpected consequence of the new generator described above. These simulations show how the results from the original work on AlChemy depend on the details of how the random expressions are generated, and we speculate that the different standardization approaches may be the key discriminating factor.

\section{$\lambda$ Expressions Simulate Chemical Reaction Networks (CRNs)}
\label{sec:simulation}

A sequence of chemical reactions can be thought of as repeated application of combinatorial substitution rules to sets of molecules (the reagents), producing a set of products, an approach that is formalized in CRNs~\cite{andersen2016software, feinberg2019foundations}. CRNs are used widely in models of the origin of life~\cite{oolen2023takes}, biochemical systems~\cite{endy2001modelling}, complex systems science~\cite{borisov2021two}, and dynamical systems theory~\cite{gambino2013turing}.  In this section we establish a formal relationship between CRNs and $\lambda$ calculus by showing that for any given CRN, the transformations it enables can be simulated by a set of $\lambda$ expressions and appropriate reductions. 

At first glance, a system based on $\lambda$ calculus seems very far removed from chemistry. The main mathematical tool used to study the dynamics of chemical systems are Chemical Reactions Networks (CRNs). Both CRNs and pure $\lambda$ calculus can simulate Turing machines, and therefore are Turing complete. This alone does not mean that $\lambda$ calculus is a reasonable model of chemistry, especially considering that the rules of AlChemy are more than the rules of $\lambda$ calculus alone: they also include rules about collisions and constraints due to pragmatic reduction. \textcolor{black}{Here we show that a given CRN's state transitions can be simulated in a system like AlChemy (albeit based on typed $\lambda$ calculus). Our proof is more direct than simply stating the computational equivalence of all Turing complete systems because it provides the construction. In particular, the construction identifies a $\lambda$ expression for each individual reagent in the CRN and uses the collision rules of AlChemy to replicate the state transitions of the CRN. %collide these expressions together.
Our construction relies on typed $\lambda$ calculus, and it does not address the likelihood that any given AlChemy run (with random collisions) will simulate a given CRN, the time dependent behavior of the CRN, or the reaction rates of the CRN. Instead, it shows that there exists a run that simulates the CRN's behavior. }% 

Although easy to describe informally, and widely used in chemistry~\cite{feinberg2019foundations}, we next define CRNs formally to set up the proof that $\lambda$ expressions can simulate CRNs. A CRN is defined as a set of reaction rules, represented as a tuple of two sets, $(R, S)$, where $S$ is a finite universe of chemical reagent symbols and $R$ is a set of {\it reactions}. Let $A_i, B_j \in S$ be some (not necessarily distinct) reactant and product species. A {\it reaction} is a transformation rule on the chemical species with the following form:
\[
    R_i = A_1 + \dots + A_n \longrightarrow B_1 + \dots + B_m.
\]

The state of a CRN at time $t$ is denoted $\sigma(t)$ and is a multiset whose members are symbols in the set of possible states, $S$. A multiplicity of each symbol in $\sigma(t)$ is interpreted as the quantity of the reagent corresponding to that symbol that is available at time $t$. CRN rules are applied to $\sigma(t)$ to transform it to $\sigma(t+\tau)$. The transformation for each rule $R_i$ is defined as follows: if $A_1, \ldots A_n \in \sigma$, then $R(\sigma(t)) = \sigma(t+\tau) = \sigma(t) \cup \{B_1, \ldots B_m\} \setminus \{A_2, \ldots, A_n\}$. To reiterate, each member of the reagents, $A_i$, and each member of the products, $B_i$, can appear multiple times in each rule. The multiplicity of each member is updated accordingly. If multiple rules are applicable to $\sigma(t)$, then with uniform probability one the rules is selected randomly.

If a sequence of rules $R^\star = R_{i_1}, \ldots R_{i_k}$ is applied in succession to $\sigma(t)$ to produce $\sigma(t + k)$, we write $\sigma(t) \to^\star_R \sigma(t+k)$. If there exists a sequence of rules $R^\star$ that produces some $\sigma(k)$ from $\sigma(0)$ such that $x \in \sigma(k)$, then we say that the CRN produces product(s) $x$.

We now show by construction that there exists a simply typed $\lambda$ calculus system that can simulate arbitrary sequences of state transitions of any given CRN. This is a refinement of AlChemy, since AlChemy uses untyped-$\lambda$ calculus. In a system of typed $\lambda$ expressions, collisions occur only when the expressions have compatible types, i.e., reactfions can only occur between certain expressions. In the simply typed $\lambda$ calculus, each expression is associated with a discrete {\it type}, and the types restrict the reductions allowed by untyped application. This is similar to the use of types in modern typed computer programming languages, where function evaluation requires compatible types. We introduce types to guard against stray reactions (e.g., between species whose reactions are not specified by the CRN). Interestingly in chemical systems no such formal guard exists, electrons do not have prior knowledge of which reactions they are or are not allowed to participate in. The typed $\lambda$ calculus is a computationally weaker structure than untyped $\lambda$ calculus and can be simulated by it. Although we do not give a proof here, AlChemy with types is no more difficult to simulate than AlChemy without types.

Let $C$ be a CRN and $\rho$ an AlChemy system with types. Consider a sequence of reactions $R^\star = R_{i_1}, \ldots R_{i_k}$ applied to $C$ at time $t$. We show that there exists some transformation ${\sf A}: C \to \rho$ such that if there exists a sequence of reactions $R^\star$ on $\sigma(t)$ producing $\sigma(t + k)$, then there also exists a sequence of collisions $\beta^\star$ on ${\sf A}(\sigma(t))$ that results in ${\sf A}(\sigma(t+k))$  (figure~\ref{fig:sim-comm}). For this commutativity property to hold (Figure~\ref{fig:sim-comm}), we do not need to make claims about the number of $\beta$ reductions required in $\beta*$, or even that such a sequence of reductions is likely, only that it is possible. We next demonstrate by construction how each reaction rule can be implemented. Our construction introduces a new $\lambda$ term ($L_i)$ for each rule. 
    
\begin{figure}[h]
    \centering
    \begin{tikzcd}
    	{\sigma(t)} && {\sigma(t+k)} \\
    	\\
    	{\rho_1} && {\rho_2}
    	\arrow["\mathsf{A}", from=1-3, to=3-3]
    	\arrow["\mathsf{A}", from=1-1, to=3-1]
    	\arrow["{R^\star}"{description}, from=1-1, to=1-3]
    	\arrow["{\beta^\star}"{description}, from=3-1, to=3-3]
    \end{tikzcd}
   \caption{{\bf Correspondence between $\lambda$ calculus and Chemical Reaction Networks (CRNs).} $\sigma(i)$ are states of a CRN, and $\rho_i$ are states of a typed AlChemy system.} 
    \label{fig:sim-comm}
\end{figure}

The transformation $\textsf{A}$ applies to a CRN state $\sigma(t)$ to produce a new $\lambda$ expression ($L_i$) for each available rule ($R_i$) and makes the initial set of $\lambda$ expressions $\rho_1$ with these $L_i$s. Further, for each available chemical species with count $k$ in $C$,  $k$ new $\lambda$ expressions representing those species (e.g. the $A_i$, and $B_i$) are added to $\rho_1$. The expression $L_i$ acts like a catalyst, $L_i$ not destroyed by the repeated $\beta$-reductions, and all reductions involving $L_i$ consume a set of $\lambda$ expressions that correspond to species on the left-hand side of $R_i$, and produce a set of expressions that correspond to species on the right-hand side of $R_i$.

\begin{theorem}[Simulation]
    For all CRNs $C = (R, S)$, there exists a mapping $\mathsf{A}: C \to \rho$ such that if $\sigma(t) \to_R^\star \sigma(t+k)$, then there exists some sequence of collisions $\beta^\star$ such that $\rho_1 = \mathsf{A}(\sigma(t)) \to_\beta^\star \mathsf{A}(\sigma(t+k)) = \rho_2$.
\end{theorem}

\begin{proof}
    To construct $\mathsf{A}$, we must construct a set $\rho$ of typed $\lambda$ calculus expressions for every state $\sigma$ of the given CRN $C$. The transformation proceeds as follows: For each chemical symbol $x \in S$, we introduce the type $\tau_x$ into $\rho$. For each copy of a symbol $x$ in $\sigma$, we add a value $v_x$ of type $\tau_x$ to $\rho$ . The number of $v_x$ values equals the count of $x$ in $\sigma$. For each rule (reaction) $R_i$ in the CRN, we add a single copy of the $\lambda$ term $L_i$ to $\rho$. $L_i$ is constructed such that it accepts the variables corresponding to the left-hand side of $R_i$ as arguments, and contains in its body a list of the variables corresponding to the right-hand side of $R_i$ as output:
    \begin{equation*}
        L_i = \lambda v_{A_1}:\tau_{A_1}.~\lambda v_{A_2}:\tau_{A_2}.~\ldots~\lambda v_{A_n}:\tau_{A_n}.~(v_{B_1} (v_{B_2} (v_{B_3} \dots v_{B_m}) \dots )).
    \end{equation*}
    We also introduce into $\rho$ arbitrarily many copies of the $\lambda$ expressions $T = \lambda x.~\lambda y.~x$ and $F = \lambda x.~\lambda y.~y$. These expressions represent the boolean values {\it true} and {\it false} respectively.

    The proof proceeds by induction on $k$, the number of reaction steps the CRN required to transform $\sigma(t)$ to $\sigma(t+k)$. Let $\rho_0 = \mathsf{A}(\sigma(0))$. If $x \in \sigma(t)$, then there exists some $\rho_1$ such that $\rho_0 \to_\beta^\star \rho_1$ and $v_x \in \rho_1$. If there is a sequence of reactions ($R^\star$) that transforms $\sigma(t)$ to $\sigma(t+k)$, then we must show that there exists a sequence of reactions ($\beta^\star$) transforming $\rho_1$ to $\rho_2$.
    
    By induction, we assume that each of $v_{A_1} \ldots v_{A_n}$ are all in $\sigma(t)$. Then, an application of the rule $L_i$ to each $v_{A_i}$ in succession (discarding all partial products) produces the body of $L_i$. That is,
    \begin{equation*}
        (L_i~v_{A_1}~\ldots~v_{A_n}) \to_\beta^\star (v_{B_1} (v_{B_2} (v_{B_3} \dots v_{B_m}) \dots )).
    \end{equation*}
    In AlChemy, both reactants (in this case $L_i$) are returned to the population, in addition to the product and a random expression is removed. In our reduction we assume $L_i$ is not the random expression chosen to be removed, thereby maintaining the population of rules in the reaction set.
    
    Once the list $(v_{B_1} (v_{B_2} (v_{B_3} \dots v_{B_m}) \dots ))$ is produced, it remains to decompose the list into a desired set of outputs. Using the $\lambda$ expressions $T$ and $F$, we can successively apply these expressions to the list such that we produce each $v_{B_x}$ for any chosen $x$. Note that the application of $T$ produces the first element of the list, and the application of $F$ produces the remainder of the list, a partial product. Upon each application, the partial product is returned to the soup to generate the next value. The resulting set of expressions removes $v_{A_1} \ldots v_{A_n}$ from $\rho_1$, and adds $v_{B_1} \ldots v_{B_m}$ to $\rho_1$, thereby simulating one reaction step of the CRN. 
\end{proof}

This proof demonstrates that for {\it any} set of reaction rules, we can construct a set of $\lambda$ expressions, which when composed in a correct order, has the effect of `firing' a set of reaction rules (thereby converting a set of reactants to a set of products). \textcolor{black}{This means that a carefully selected set of $\lambda$ expressions can simulate the transition of any set of reactants under any sequence of reaction rules $R$ between two different points in time (i.e. simulating the sequence $\sigma(t) \to_R^\star \sigma(t + k)$).} Alternatively, a proof that uses the computational equivalence of Turing-complete systems would construct a single lambda expression to represent an entire CRN, thereby sidestepping the dynamics of chemical reactions as lambda-expression collisions altogether. Our proof here, does not make claims about whether or not those combinations of $\lambda$ expressions or $\beta$ reductions are likely---the expressions are not random, and the given sequence of $\beta$ reductions is only one out of many other possibilities if application and subsequent reduction happen in random order. Despite these caveats, the proof demonstrates a clear formal correspondence between CRNs and typed $\lambda$ calculus, and \textcolor{black}{it illustrates how the space of AlChemy simulations contains state transformations between states of any CRN.}

\section{Discussion}

\subsection{AlChemy as a model of chemistry}

\textcolor{black}{Section~\ref{sec:simulation} shows that systems of typed $\lambda$ calculus can represent arbitrary sequences of state transitions of CRNs, which establishes a mapping between the typed $\lambda$ calculus and CRNs. We do not demonstrate that the $\lambda$ calculus captures all the richness of chemistry, including energetic, thermodynamic and structural features; nor do we attempt to replicate the exact time dependent concentrations.
However, it provides an initial step towards such a proof by demonstrating that every reaction sequence that a CRN might produce can also be produced by a corresponding set of $\lambda$ expressions.} The proof relies on typed, rather than un-typed, $\lambda$ calculus, because with un-typed $\lambda$ calculus, it would be challenging to prevent undesirable reactions (or unanticipated) reactions from occurring. The type system provides a shortcut for deciding \textit{a priori} what is and is not possible. 

It is interesting that synthetic chemists have to solve this problem manually without invoking a type system, and instead create a new one. They do this by using pure reagents, and then identifying compounds with limited and specific reactivity which guide their compounds through specific desired reaction rules, and no others. Thus, CRNs themselves are a refined mathematical abstraction, which can be used to make predictions about the systems that synthetic chemists (or other biological entities) have made. They include only pre-specified reactions and rate constants, which can be used to derive the dynamical properties of the system. Rate constants and reaction rules can reflect energetic, or thermodynamic, considerations, but they are additional features layered on top of the original CRN model and are not endogenous to it.

It may be possible to achieve similar features in an AlChemy like system. For example, differential reaction rates could be incorporated using typed $\lambda$ calculus by introducing {\it solvent} expressions, which can only bind and unbind to specific reactants, thereby modulating what they can react with. By using a diversity of different solvent expressions (which need not correspond to different solvents), with different types, it would be possible to modulate the relative rates of reactions. From a chemical perspective this would be analogous to solvation effects. Driving the analogy further, it might be possible to draw a correspondence between reaction progress and the number of $\beta$ reductions required to reach normal form, with very long reduction processes implying elaborate molecular transformations. Simulating chemical systems involves a diversity of different modeling approaches itself, from first-principles quantum chemical calculations, thermodynamic models, to spatially embedded reaction diffusion systems. It seems unlikely that a single $\lambda$ calculus model would subsume all of these approaches elegantly and simultaneously.

\subsection{AlChemy as a model of computation}

An appealing feature of {\textit AlChemy} is the emergence of complex stable organizations. These organizations take different forms, often consist of different individual expressions, and manifest different relational architectures. All of this emerges from random $\lambda$ expressions, which can be understood as random computational functions. A natural question to ask is whether a bottom-up process like AlChemy can generate anything useful, e.g., perform computations that a software engineer would find interesting. Similarly, if AlChemy were seeded with functions known to be useful in other contexts, what would emerge?  For example, program synthesis is an active subfield of computer science, and it would be interesting to ask how an AlChemy-like system would behave if it were seeded with functions written in a modern functional programming language like Haskell and whether it might produce (synthesize) useful new functionality.  Haskell is a natural choice for such a project, because it has a large repertoire of existing functions that could be used to populate the soup.  A bottom-up randomly-driven synthesis process for code would be a radical departure from existing practice and might or might not produce compelling results.  Another approach might involve combining AlChemy's bottom-up behavior with evolutionary computation, whether in the realm of software improvement or some other domain.  Such an approach would be reminiscent of novelty search~\cite{novelty-search} and could help address well-known problems in evolutionary computation that arise with objective fitness functions and lead to premature convergence.  If randomly assembled functions from a simulation like AlChemy can be searched and their performance measured against data (say, test cases), it could provide a new approach for many problems and would represent an extension of bio-inspired computing to include inspiration from chemical systems, particularly the chemical processes that led to life on Earth.

\subsection{Origin of Life and Artificial Life}

How might one detect evolution via selection in AlChemy? Our analysis of L2 organizations shows that interactions between L1 organizations are non-trivial, often generating dynamical properties that do not admit simple analysis. Recently, Assembly theory~\cite{marshall2021identifying, sharma2023assembly} has been proposed for studying selection in non-biological systems, particularly in vast combinatorial systems, such as chemistry. Assembly theory is specific enough to make predictions about empirical systems, but general enough to be useful as a theoretical framework, and it could be useful for characterizing selection in AlChemy. By tracking the Assembly Index and the Assembly Space of $\lambda$ expressions through time in simulations, one might be able to detect selection and characterize the functional motifs driving that selection. This requires understanding $\lambda$ expressions as composite objects which can be decomposed into elementary building blocks. In the case of molecules, this is relatively straightforward, and doing so has provided an experimental technique for detecting life and a method for generating novel drug-like compounds for drug discovery~\cite{marshall2021identifying,liu2021exploring}. Mapping the assembly space of expressions in AlChemy simulations could provide a way to characterize the accessibility of this function space. 

Our results using different random expression generators illustrate an important point about the emergence of stable organizations: the richness of the organization depends on the accessibility of the space of objects. In our simulations the space of possible objects is completely determined by the initial conditions and those initial conditions are determined by the properties of the random expression generator. When we used the permutation generator for initial conditions (Section~\ref{sec:random-exprs}), the simulation produced $\lambda$ expressions that were {\it too} reactive to form stable organizations. Expressions interacted rapidly, those reactions generated a dense network of other reactions, so the system could explore the entire state-space and thus collapsed a trivial fixed point. The stability of this fixed point was not an autocatalytic property, but rather it arose because the expressions it contained were inert. With these initial conditions, the sequences of reactions between objects had many possible routes to produce an inert object that could not react further. Meanwhile, original expression generators led to much greater diversity of organizations, and it was rare for the simulations to end in the trivial fixed point (Figure \ref{fig:L0sims}). In this case, although the expressions could react with each other, the reactions rarely produced inert expressions. This meant that future reactions were always possible, ensuring the dynamic stability of the organization. These two types of stability (inert fixed point, and dynamic autocatalytic stability) can be analogized to thermodynamic stability, and the concept of ``dynamic-kinetic stability''\cite{pross2011toward}. This implies that a simple way to drive the spontaneous emergence of dynamically stable organizations is by designing systems that inhibit a direct approach to inert states. Interestingly, the organic chemistry of life on Earth was possible because organic carbon can ``get stuck" between the relatively inert redox states of carbon dioxide and methane \cite{shock2015principles, falkowski2008microbial, smith2004universality}. An interesting question, for both AlChemy and organic chemistry, is the extent to which Assembly theory can determine the connectivity of the space of allowable transformations. AlChemy would represent a useful test case for evaluating whether one can predict the accessible futures of a constructive system from the Assembly space of its initial conditions alone.

\section{Conclusion}

Although AlChemy was successful in many ways, the origin of life is an still unresolved question. Yet, experimental progress towards understanding the chemical origins of life is demonstrating that simple, unconstrained and nearly random systems can be guided towards diverse, yet constrained, product spaces~\cite{robinson2022environmental, muchowska2020nonenzymatic, surman2019environmental, asche2021robotic}. As experimental access to the full richness of these systems improves, the principles underlying the emergence of life-like phenomena in chemistry will become clearer~\cite{cronin2016beyond, jirasek2023multimodal}. New generations of scientists will seek to harness those principles for design and engineering aspirations, as we witnessed with the principles underlying biology~\cite{preiner2020future, miikkulainen2021biological}. Simulations like AlChemy can serve as a computational testbed for those principles, enabling proof of concept simulations which can guide new experimental paradigms and suggest targets for analysis to track life-like features of physical systems.

\section*{}
\subsection*{Acknowledgments}
CM and DP would like to thank Dr. Doug Moore for his help compiling the original AlChemy software, and Dr. Harrison B Smith, Veronica Mierzejewski, and Gage Siebert for their critical comments on an earlier version of this manuscript.  SF acknowledges the partial support of NSF (CCF2211750, CICI 2115075), DARPA (FA8750-19C-0003, N6600120C4020), ARPA-H (SP4701-23-C-0074), and the Santa Fe Institute.

\subsection*{Author Contributions}
All authors contributed to equally to project design, evaluation of
results, and preparation of the manuscript.  CM and DP developed the
version of AlChemy that runs in modern computing environments; CM ran
simulations; and DP performed theoretical analyses. CM and SF coordinated the research effort, and SF provided research funding. 

\bibliographystyle{unsrt} 
\bibliography{ref}

\end{document}